\RequirePackage{fix-cm}
\documentclass[smallextended]{svjour3}       
\smartqed  

\spnewtheorem{definitionitalic}{Definition}{\bf}{\it}
\spnewtheorem*{nonumtheorem}{Theorem}{\bf}{\it}
\spnewtheorem{assumption}{Assumption}{\bf}{\it}

\usepackage{bbm}
\usepackage{bm}
\usepackage{amsmath,amssymb}
\usepackage{verbatim}
\usepackage{stmaryrd}
\usepackage{graphicx,color}
\usepackage{mathtools}
\usepackage{slashed}
\usepackage[all,cmtip]{xy}
\usepackage{tikz-cd}
\usepackage{array}
\usepackage{enumerate}

\DeclareMathOperator{\coker}{coker}

\DeclareMathOperator{\Gr}{Gr}

\DeclareMathOperator{\rank}{rank}
\DeclareMathOperator{\Der}{Der}

\DeclareMathOperator{\CDerEnd}{CDer}

\DeclareMathOperator{\image}{image}

\DeclareMathOperator{\Hom}{Hom}

\DeclareMathOperator{\End}{End}

\newcommand{\eps}{\varepsilon}

\newcommand{\p}{\partial}

\renewcommand{\epsilon}{USE eps INSTEAD}

\newcommand{\R}{\mathbbm{R}}
\newcommand{\Z}{\mathbbm{Z}}

\renewcommand{\subset}{\subseteq}
\renewcommand{\supset}{\supseteq}
\newcommand{\secref}[1]{Section \ref{sec:#1}}

\newcommand{\lemmaref}[1]{Lemma \ref{lemma:#1}}
\newcommand{\theoremref}[1]{Theorem \ref{theorem:#1}}

\newcommand{\defref}[1]{Definition \ref{def:#1}}

\newcommand{\catstyle}[1]{\textnormal{\textsf{{\fontseries{sbc}\selectfont #1}}}}

\newcommand{\E}{\mathcal{E}}

\newcommand{\MC}{\catstyle{MC}}

\newcommand{\RR}{{C^\infty}}

\newcommand{\px}{\mathcal{P}}
\newcommand{\ax}{\mathcal{A}}

\newcommand{\sbounce}{\textnormal{bounce}}
\newcommand{\sfree}{\textnormal{free}}
\newcommand{\pbounce}{\px_{\sbounce}}
\newcommand{\pfree}{\px_{\sfree}}
\newcommand{\abounce}{\ax_{\sbounce}}
\newcommand{\afree}{\ax_{\sfree}}

\newcommand{\inj}{\hookrightarrow}
\newcommand{\surj}{\twoheadrightarrow}

\newcommand{\sv}[1]{\mathbf{s}_{#1}}

\newcommand{\transition}{\chi}

\newcommand{\citeone}{\cite{rt1}}
\newcommand{\citetwo}{\cite{rt2}}

\newcommand{\ddc}[1]{d_{#1}}

\newcommand{\oldw}{K}

\newcommand{\Donezero}{D_{10}}
\newcommand{\Dthreetwo}{D_{32}}
\newcommand{\Dfourthree}{D_{43}}
\newcommand{\slashedDfourthree}{\slashed{D}_{43}}
\newcommand{\Dthreeone}{\slashed{D}_{3}}
\newcommand{\spatialm}{S}

\usepackage{mathptmx}      
\begin{document}

\title{Filtered expansions in general relativity II}


\author{Michael Reiterer \and
        Eugene Trubowitz
}


\institute{Michael Reiterer,
              ETH Zurich,
              \email{michael.reiterer@protonmail.com}             \\
            \rule{55pt}{0pt} \text{Present affiliation: The Hebrew University of Jerusalem}\\
           Eugene Trubowitz,
              ETH Zurich,
              \email{eugene.trubowitz@math.ethz.ch}
}

\date{Received: date / Accepted: date}

\maketitle

\begin{abstract}
This is the second of two papers
in which we construct formal power series solutions
in external parameters
to the vacuum Einstein equations,
implementing one bounce for the Belinskii-Khalatnikov-Lifshitz (BKL) proposal
for spatially inhomogeneous spacetimes.
Here we show that spatially inhomogeneous perturbations of spatially
homogeneous elements are unobstructed.
A spectral sequence for a filtered complex,
and a homological contraction based on gauge-fixing, are used to do this.

\keywords{First keyword \and Second keyword \and More}
\end{abstract}

\section{Introduction}

This is a continuation of \citetwo,
using the homological,
graded Lie algebra (gLa) framework for general relativity in \citeone.
Familiarity with both these papers is assumed,
this introduction only gives pointers to the pertinent material.

In {\citeone} we defined a gLa $\E$
whose nondegenerate Maurer-Cartan (MC) elements
are the solutions to the vacuum Einstein equations. 
In {\citetwo} we first defined what a filtered expansion is in general,
then introduced a specific 1-index gLa filtration called BKL filtration,
and derived 2- and 3-index filtrations from it.
Their Rees algebras
    \begin{align*}
      \pbounce & \;=\;
      \{ \textstyle\sum_{p_2,p_3} s^{p_2}_2 s^{p_3}_3 x_{p_2p_3}
      \mid x_{p_2p_3} \in F_{p_2p_3}\E\}\\
      \pfree & \;=\;
      \{ \textstyle\sum_{p_1,p_2,p_3}
      s^{p_1}_1 s^{p_2}_2 s^{p_3}_3 x_{p_1p_2p_3}
      \mid x_{p_1p_2p_3} \in F_{p_1p_2p_3}\E\}
\end{align*}
are free over $\R[[s_2,s_3]]$ and
      $\R[[s_1,s_2,s_3]]$,
  subalgebras of $\E[[s_2,s_3]]$ and $\E[[s_1,s_2,s_3]]$ respectively.
  Here $s_1,s_2,s_3$ are formal symbols.
The associated graded gLa are
\[
  \abounce = \pbounce/(s_2,s_3)
  \qquad
  \afree = \pfree/(s_1,s_2,s_3)
\]

In this paper we take certain MC-elements in $\abounce$ and $\afree$
and calculate the homology of the associated differentials.
An essential point is that we 
study these differentials on the associated gradeds, $\abounce$ and $\afree$,
because they control formal perturbations on $\pbounce$ and $\pfree$,
which is what one is actually interested in.
See {\citeone}, {\citetwo}.
With this understanding, this paper is set in the associated gradeds.

In {\citetwo} we constructed quite general MC-elements in $\abounce$ and $\afree$.
They depend on several function parameters,
and so do their differentials.
\emph{We only study the
differential associated to spatially homogeneous MC-elements.
These differentials control spatially inhomogeneous perturbations,
not merely homogeneous ones.}
This keeps calculations manageable.
We speculate that results would be similar starting from
the spatially inhomogeneous MC-elements in \citetwo, but we have not checked this.

The differential on the associated gradeds defines a filtered complex that
is studied using two tools:
a spectral sequence for a filtered complex, reviewed in {\citetwo};
and a contraction based on gauge-fixing {\citeone},
a tool specific to general relativity.

\begin{nonumtheorem}[No obstructions - informal version]
  For homogeneous MC-elements in $\afree$ and $\abounce$,
  the associated differential has vanishing second homology.
\end{nonumtheorem}
Hence by {\citetwo}, spatially inhomogeneous
formal power series `free motion' and `bounce' solutions exist,
implementing building blocks of the BKL proposal.
We first suggested in \cite{fil} that these obstructions could vanish.

The main ingredients from {\citetwo} are:
\begin{enumerate}[(J1)]
  \item \label{dep:npar} The over-parametrization lemma \cite[Lemma 8]{rt2}.
  \item \label{dep:mcefreehom} The remark about homogeneous MC-elements in $\afree$ \cite[Remark 2]{rt2}.
  \item \label{dep:bouncesol} The theorem about MC-elements in $\abounce$ \cite[Theorem 3]{rt2}.
\end{enumerate}
\newcommand{\depref}[1]{\textnormal{(J\ref{dep:#1})}}

\newcommand{\spx}{\,\,\,}
 \begin{table}
   \footnotesize{
   \[
   \begin{array}{|l|l|l|}
     \hline
     \alpha = p_1p_2p_3 & \text{$G_\alpha \E_G \subset \E$ is the $\RR$-span of these elements} & \text{$\RR$-rank} \\
     & 
     \text{note that
     $G_\alpha \E = G_\alpha \E_G \oplus \theta_0(G_\alpha \E_G)$
     }
     &\\
     \hline
     \hline
     000 & \Der(\RR),\spx\sigma_0,\spx\theta_0 \sigma_0 + \theta_1\sigma_1,\spx
             \theta_0\sigma_0 + \theta_2\sigma_2,\spx \theta_0\sigma_0 + \theta_3\sigma_3, & 9\\
             & 
             \theta_2\theta_3\sigma_{23}+\theta_3\theta_1\sigma_{31}+\theta_1\theta_2\sigma_{12}
             + 2\theta_0\theta_1\sigma_1+2\theta_0\theta_2\sigma_2+2\theta_0\theta_3\sigma_3 & \\
     \hline
     200 & -\theta_1\sigma_{23} + \theta_2\sigma_{31} + \theta_3\sigma_{12} & 1\\
     \hline
     020 & +\theta_1\sigma_{23} - \theta_2\sigma_{31} + \theta_3\sigma_{12} & 1\\     
     \hline
     002 & +\theta_1\sigma_{23} + \theta_2\sigma_{31} - \theta_3\sigma_{12} & 1\\
     \hline
     011 & \sigma_1,\spx
             \sigma_{23},\spx
             \theta_1 \Der(\RR),\spx
             \theta_0\sigma_1+\theta_1\sigma_0,\spx
             \theta_2\sigma_3+\theta_3\sigma_2, & 11\\
           & \theta_3\sigma_{31},\spx
             \theta_2\sigma_{12},\spx
             \theta_0\theta_2\sigma_{12}+\theta_1\theta_2\sigma_2 & \\
             \hline
     101 & \sigma_2,\spx
             \sigma_{31},\spx
             \theta_2 \Der(\RR),\spx
             \theta_0\sigma_2+\theta_2\sigma_0,\spx
             \theta_3\sigma_1+\theta_1\sigma_3, & 11\\
           & \theta_1\sigma_{12},\spx
             \theta_3\sigma_{23},\spx
             \theta_0\theta_3\sigma_{23}+\theta_2\theta_3\sigma_3 & \\
             \hline
     110 & \sigma_3,\spx
             \sigma_{12},\spx
             \theta_3 \Der(\RR),\spx
             \theta_0\sigma_3+\theta_3\sigma_0,\spx
             \theta_1\sigma_2+\theta_2\sigma_1, & 11\\
           & \theta_2\sigma_{23},\spx
             \theta_1\sigma_{31},\spx
             \theta_0\theta_1\sigma_{31}+\theta_3\theta_1\sigma_1 & \\
             \hline
     211 & \theta_2\sigma_3 - \theta_3\sigma_2,\spx
             \theta_2\theta_3 \Der(\RR), & 7\\
           & \theta_0\theta_2\sigma_3-\theta_0\theta_3\sigma_2-2\theta_2\theta_3\sigma_0,\spx
             \theta_0\theta_2\sigma_3+\theta_0\theta_3\sigma_2-2\theta_1\theta_2\sigma_{31} & \\
             \hline
     121 & \theta_3\sigma_1 - \theta_1\sigma_3,\spx
             \theta_3\theta_1 \Der(\RR), & 7\\
           & \theta_0\theta_3\sigma_1-\theta_0\theta_1\sigma_3-2\theta_3\theta_1\sigma_0,\spx
             \theta_0\theta_3\sigma_1+\theta_0\theta_1\sigma_3-2\theta_2\theta_3\sigma_{12} & \\
             \hline
     112 & \theta_1\sigma_2 - \theta_2\sigma_1,\spx
             \theta_1\theta_2 \Der(\RR), & 7\\
           & \theta_0\theta_1\sigma_2-\theta_0\theta_2\sigma_1-2\theta_1\theta_2\sigma_0,\spx
             \theta_0\theta_1\sigma_2+\theta_0\theta_2\sigma_1-2\theta_3\theta_1\sigma_{23} & \\
             \hline
     222 & \theta_0\theta_1\sigma_{23}+\theta_0\theta_2\sigma_{31}
             + \theta_0\theta_3\sigma_{12}
             - 2 \theta_2\theta_3\sigma_1
             - 2 \theta_3\theta_1\sigma_2
             - 2 \theta_1\theta_2\sigma_3, & 6 \\
           & \theta_1\theta_2\theta_3 \Der(\RR), & \\
           & \theta_0\theta_2\theta_3\sigma_1
            +\theta_0\theta_3\theta_1\sigma_2
            +\theta_0\theta_1\theta_2\sigma_3
            +3\theta_1\theta_2\theta_3\sigma_0 &
     \\
     \hline
     \text{else} & \text{none} & 0\\
     \hline
 \end{array}
 \]}
 \caption{
 Definition of $G_\alpha\E_G$
 associated to a 
 conformally orthonormal basis $\theta_0,\theta_1,\theta_2,\theta_3$.
 We omit wedges, so $\theta_0\theta_1 = \theta_0 \wedge \theta_1$.
 The basis-dependent injection $\Der(\RR) \inj \CDerEnd(W)$ is implicit \citeone.
 All elements in the table are
 elements of $\E$ via the canonical surjection $\mathcal{L} \surj \E$.
 }\label{table:EG}
 \end{table}

\section{The homology in $\afree$\\
of some homogeneous Maurer-Cartan elements}\label{sec:homology}
Let $\gamma \in \MC(\afree)$ be a homogeneous MC-element
as in Lemma \depref{mcefreehom}.
Here we compute the homology $H^k(d)$ of the differential
\[
  d = [\gamma,-]\;\;\in\;\; \End^1(\afree)
\]

\begin{assumption}[In force through \secref{homology}] \label{assp:mc}
Suppose $\gamma \in \MC(\afree)$ as above and:
  \begin{itemize}
    \item $g_1^0$, $g_2^0$, $g_3^0$
      are nonzero and pairwise different. The sum of any two is negative.
    \item $g_1^1$, $g_2^1$, $g_3^1$ are nonzero. 
    \item $\spatialm$ is an open subset of a smooth 3-dimensional Lie group
      containing the identity element,
      and $D_1$, $D_2$, $D_3$ is a basis of left-invariant vector fields.
    \item There is an $\eps > 0$ and a diffeomorphism
      $(x^1,x^2,x^3) : \spatialm \to (-\eps,\eps)^3$
      such that $(0,0,0)$ is the identity element and,
      with $\p_1,\p_2,\p_3$ the partial
      derivatives for the coordinate system $x^1,x^2,x^3$,
      one has\footnote{%
        Let $\phi^i$ be the flow associated to $D_i$
        on the Lie group; this flow exists.
        Define the map $(x^1,x^2,x^3) \mapsto (\phi^2_{x^2} \circ 
        \phi^1_{x^1} \circ \phi^3_{x^3})(e)$ where $e$ is the identity element.
        Using the inverse function theorem, we can construct a coordinate system
        with the required properties.
      }
      \begin{equation}\label{eq:d2cond}
        D_2 = \p_2
        \qquad
        D_1|_{x^2 = 0} = \p_1
      \end{equation}
  \end{itemize}
  \end{assumption}
  The last assumption is technical.
  An alternative would be to work not over $C^{\infty}(\spatialm,\R)$
  but over germs of smooth functions at the identity
  element of the Lie group.
  \begin{theorem}[No obstructions] \label{theorem:noobstr}
    With Assumption \ref{assp:mc}
    we have:
    \begin{itemize}
      \item $H^0(d) \simeq \R^3$, the right-invariant vector fields.
      \item $\Gr H^1(d) \simeq H^1(D_{00}) \oplus \ker \Dthreeone$,
        see Lemmas \ref{lemma:hom0page} and \ref{lemma:kerd13}.
      \item $H^2(d) = 0$.
      \item $H^3(d) = 0$.
      \item $H^4(d) = 0$.
    \end{itemize}
    Here $\simeq$ is an isomorphism as vector spaces,
    and $\Gr$ is the associated graded for the decreasing filtration
    coming from
    the $\Z$-grading of $\afree$ by $p_2 + p_3$.
  \end{theorem}
  \begin{proof}
    This follows from the remainder of \secref{homology},
    and the fact {\citetwo}
    that the last page of the spectral sequence is isomorphic
    to the associated graded of the homology.
    For example, this way we get $\Gr H^2(d) = 0$, which implies $H^2(d)=0$.
  \qed\end{proof}
\subsection{Overview}
Recall {\citetwo} that $\afree \simeq \mathcal{U}$
is graded by tuples $\alpha = p_1p_2p_3$, see Table \ref{table:EG}.
There is a corresponding decreasing filtration
$F_{\geq \alpha} \afree \simeq F_{\geq \alpha} \mathcal{U}$ given by
\[
  F_{\geq \alpha} \mathcal{U} \;=\; \textstyle\bigoplus_{\beta \geq \alpha} G_{\beta} \mathcal{E}
\]
It is respected by the differential,
so we have a decreasingly filtered complex.
That is, we have
    $d(F_{\geq\alpha} \afree)\subset
      F_{\geq\alpha}\afree$.
The homology is calculated in two steps:
\begin{itemize}
  \item Via gauge fixing,
    we construct a contraction of $(\afree,d)$
    onto a filtered complex $(C,d_C)$,
    so they are quasi-isomorphic.
    The space $C$ 
    is a free $C^{\infty}(\spatialm,\R)$-module.
    The contraction is a variant of the one in {\citeone}.
    Since $C^4 = 0$ we get $H^4(d)=0$.
  \item We use a spectral sequence to compute
    the homology of $(C,d_C)$.
    The spectral sequence is relative to the decreasing filtration of $C$ associated
    to the single index $p_2+p_3$.
    We found this to be more useful than
    the total index $p_1+p_2+p_3$.
\end{itemize}

\newcommand{\jj}[1]{\stackrel{\phantom{(0)}}{#1}}
We anticipate
the structure of the spectral sequence.
In the following diagrams,
left to right corresponds to the filtration index $p_2+p_3 = 0,1,2,3,4$.
The homological degrees are suppressed;
one should visualize each bullet as being a complex
with arrows perpendicular to the page.
A $(k)$ over a bullet means that the complex
is concentrated in homological degree $k$.
We start with the 0th page, which is disconnected as always
(differentials such as $D_{ii}$ are introduced later, e.g.~\secref{0thpage}):
\[
\xymatrix{%
  \jj\bullet \ar@(dl,dr)[]_{D_{00}} &
  \jj\bullet \ar@(dl,dr)[]_{D_{11}} &
  \jj\bullet \ar@(dl,dr)[]_{D_{22}} &
  \jj\bullet \ar@(dl,dr)[]_{D_{33}} &
  \jj\bullet \ar@(dl,dr)[]_{D_{44}}
}
\]
The differentials
on the 0th page are linear over $C^{\infty}(\spatialm,\R)$,
    and all the spaces on the 0th and 1st pages are free $C^{\infty}(\spatialm,\R)$-modules.
The 1st page is:
\[
\xymatrix{%
  \stackrel{\mathclap{(0)\oplus (1)}}{\bullet} \ar[r]_{\Donezero}
  & \stackrel{(1)}{\bullet}
  & \stackrel{(1)}{\bullet} \ar[r]_{\Dthreetwo}
  & \stackrel{(2)}{\bullet} \ar[r]_{\Dfourthree}
  & \stackrel{(3)}{\bullet}
}
\]
Most complexes are concentrated in a single homological degree,
excluding some arrows on the 1st and subsequent pages
for degrees reasons. The 2nd page is:
\[
\xymatrix{%
  \stackrel{\mathclap{(0)\oplus(1)}}{\bullet}
  & \stackrel{(1)}{\bullet} \ar@/_/[rr]_{\Dthreeone}
  & \jj{0}
  & \stackrel{(2)}{\bullet}
  & \jj{0}
}
\]
There is no differential from the 3rd page on:
\[
\xymatrix{%
  \stackrel{\mathclap{(0)\oplus(1)}}{\bullet}
  & \stackrel{(1)}{\bullet}
  & \jj{0}
  & \jj{0}
  & \jj{0}
}
\]
The arrows are
matrix differential operators,
of order $\leq p$ on page $p$.
We will find $\ker \Donezero \simeq \R^3$;
that $\Dthreetwo$ has trivial kernel;
and that $\Dfourthree$ has trivial cokernel.
We then construct the matrix differential operator $\Dthreeone$,
and show that it has trivial cokernel.


\subsection{A contraction and the complex $(C,d_C)$}

Recall Assumption \ref{assp:mc}.
The complex $(\afree,d)$ lives in 4 dimensions,
namely $\afree$ is a
free $C^{\infty}(\R \times \spatialm,\R)$-module,
but there is a contraction to a complex $(C,d_C)$
that lives in 3 dimensions, with $C$ a free $C^{\infty}(\spatialm,\R)$-module.
The differentials are $\R$-linear.

Here we define $C$, give a basis
compatible with the decreasing filtration,
and define $d_C$. We do not write $d_C$
down completely because the formulas would be too long.
Luckily,
we do not need every detail about $d_C$.
\begin{lemma}[The space $C$ and basis elements] \label{lemma:Cspace}
  Define
  \[
    C \;=\; \mathcal{U}_G/t\mathcal{U}_G
    \;\simeq \;{\afree}_{,G} / t {\afree}_{,G}
  \]
  which is a free module over
  \[
        C^{\infty}(\R \times \spatialm, \R) / t C^{\infty}(\R \times \spatialm, \R)
        \;\simeq\;
        C^{\infty}(\spatialm,\R)
  \]
  of rank $72$.
  A basis is given by the elements in Table \ref{table:EG}. We denote by
  $b^{\alpha}_i$ the $i$-th element listed for $G_{\alpha}\mathcal{E}_G$,
  times $s^{\alpha}$.
  Each occurrence of $\Der(\RR)$ is replaced by the ordered list $D_0,D_1,D_2,D_3$.
  We denote by $b^{k,\alpha}_i$ the $i$-th element in $G_{\alpha} \mathcal{E}_G^k$.
  Examples:
  \begin{align*}
    b^{000}_3 & = D_2\\
    b^{011}_3 & = s_2s_3\theta_1 D_0\\
    b^{1,000}_3 & = \theta_0\sigma_0 + \theta_3 \sigma_3\\
    b^{2,121}_6 & = s_1 s_2^2 s_3
                      (\theta_0\theta_3\sigma_1
                      + \theta_0\theta_1\sigma_3 - 2\theta_2\theta_3\sigma_{12} )
  \end{align*}
  This module comes with a grading
  $C = \bigoplus_{p_1p_2p_3}G_{p_1p_2p_3}C$, and
  the corresponding decreasing filtration is denoted $F_{\geq \alpha}C$.
\end{lemma}
\begin{proof}
  Omitted.
\qed\end{proof}
\begin{lemma}[The contraction and the differential $d_C$] \label{lemma:defret}
Let $\oldw$ be the composition
\[
  \oldw\;:\; {\afree}_{,G} \hookrightarrow \afree \xrightarrow{d} \afree \twoheadrightarrow \afree/{\afree}_{,G}
\]
an $\R$-linear map. Then:
\begin{itemize}
  \item $\oldw$ has an $\R$-linear right inverse ($\oldw$ is surjective).
  \item There is a homological contraction of the complex $(\afree,d)$
onto the subcomplex $(\ker \oldw,d|_{\ker \oldw})$.
In particular, they have the same homology.
  \item The canonical map $\ker \oldw \to C$, namely the composition
    \[
        \ker \oldw
        \;\hookrightarrow\;
        {\afree}_{,G}
        \;\twoheadrightarrow\;
        {\afree}_{,G}/t{\afree}_{,G} \simeq C
    \]
    is a vector space isomorphism.
    The differential $d|_{\ker \oldw}$ induces a differential $d_C \in \End^1(C)$
    that respects the filtration,
    $d_C(F_{\geq \alpha}C) \subset F_{\geq \alpha} C$.
\end{itemize}
\end{lemma}
\begin{proof}
  The proof mimics 
  the one for the contraction in {\citeone}.
  The map $\oldw$ is surjective because, relative to a basis
  compatible with the filtration as in the table, 
  $\oldw$ is lower block triangular with each square diagonal block of the form
  $D_0\mathbbm{1} + M$ with $M$ a square matrix whose entries are functions.
  Surjectivity follows from global
  existence for linear ordinary differential equations.
  Hence $\ker \oldw \to C$ is an isomorphism,
  with $C$ the vector space of initial data.
  The surjectivity yields a contraction for abstract reasons, just as in {\citeone}.
  The argument in {\citeone} differs only in that
  it uses a symmetric hyperbolic system of PDE rather than an ODE.
\qed\end{proof}
\begin{example}
  Under the isomorphism $\ker \oldw \to C$, the preimage of
  $f b^{0,000}_3
  = f D_2$
  with $f \in C^{\infty}(\spatialm,\R)$ is
  $fD_2$
  with $f$ extended
  to an element of $C^{\infty}(\R \times \spatialm,\R)$
  by $D_0(f)=0$.
  In fact, $d (fD_2) \in {\afree}_{,G}$.
\end{example}
  \begin{example}
  We discuss one example in more detail.
  Calculations like this can be automated.
  We claim that the preimage under $\ker \oldw \to C$ of
  $f b^{0,000}_5
  = f\sigma_0$
  is
  \begin{equation}\label{eq:ttt}
      f\sigma_0
      +
      t\big(
      f D_0
      + \sv{2}\sv{3} D_1(f) \sigma_1
      + \sv{3}\sv{1} D_2(f) \sigma_2
      + \sv{1}\sv{2} D_3(f) \sigma_3
      \big)
    \end{equation}
    with $D_0(f)=0$ and $\sv{i} = s_i e^{t g_i^0}$; 
  the data in the exponent belongs to the homogeneous MC-element
  $\gamma \in \afree^1$ as in Assumption \ref{assp:mc},
  whose homology we are studying.
  Since \eqref{eq:ttt} is in ${\afree^0}_{,G}$
  and modulo $t$ yields $f\sigma_0$,
  we only have to check that applying $d$ yields zero in $\afree/{\afree}_{,G}$.
  Then this is the preimage by \lemmaref{defret}.
  The following brackets are in $\afree$,
  the entry in the first slot always coming from $\gamma$.
  \begin{itemize}
    \item $G_{000}\afree^1$: One calculates,
      \begin{multline*}
      [\theta_0D_0 - \textstyle\sum_i g_i^0 (\theta_i\sigma_i + \theta_0\sigma_0),
      f \sigma_0 + t f D_0]
      = f \textstyle\sum_i g_i^0 (\theta_i \sigma_i + \theta_0 \sigma_0)
      \;\in\; G_{000} {\afree^1}_{,G}
    \end{multline*}
  \item $G_{011}\afree^1$: One calculates,
    \begin{multline*}
      [\theta_0D_0 - \textstyle\sum_i g_i^0 (\theta_i\sigma_i + \theta_0\sigma_0),
      t \sv{2}\sv{3} D_1(f) \sigma_1]
      +
      [\sv{2}\sv{3}\theta_1D_1,f \sigma_0 + t f D_0]\\
      =
      \sv{2}\sv{3}
      \big(
      -f(1+t(g_2^0+g_3^0))\theta_1D_1
      + t D_1(f) (g_2^0 \theta_2 \sigma_{12} - g_3^0 \theta_3 \sigma_{31})\\
      \rule{20mm}{0pt}
      + D_1(f) (1+t(g_1^0+g_2^0+g_3^0)) (\theta_0\sigma_1 + \theta_1\sigma_0)
      \big)
      \;\in\; G_{011} {\afree^1}_{,G}
    \end{multline*}
  \end{itemize}
  Similar for $G_{101}$, $G_{110}$.
  Since $\afree^1/{\afree^1}_{,G}$ only has these four pieces, we are done.
\end{example}
\begin{definition}[Other gradings of $C$] \label{def:Cgradings}
  Recall $C = \bigoplus_{p_1p_2p_3} G_{p_1p_2p_3} C$.
  Define the coarser $G_{p_2p_3}C = \bigoplus_{p_1} G_{p_1p_2p_3} C$,
  explicitly
  \[
      \begin{pmatrix}
        G_{00}C & G_{01}C & G_{02}C\\
        G_{10}C & G_{11}C & G_{12}C\\
        G_{20}C & G_{21}C & G_{22}C
      \end{pmatrix}
      =
            \begin{pmatrix}
              G_{000}C\oplus G_{200}C & G_{101}C & G_{002}C\\
              G_{110}C & G_{011}C \oplus G_{211}C & G_{112}C\\
        G_{020}C & G_{121}C & G_{222}C
      \end{pmatrix}
  \]
  The ordering of the direct sums gives ordered bases,
  so the basis for $G_{00}C$ is the concatenation
  of the ordered bases for $G_{000}C$ and $G_{200}C$.
  Same for $G_{00} C^k$ for every $k$.
  We also define $V_p = \bigoplus_{p_2+p_3 = p} G_{p_2p_3}C$,
  explicitly
  \begin{equation}\label{eq:vspaces}
  \begin{aligned}
    V_0 & = G_{00}C\\
    V_1 & = G_{10}C \oplus G_{01}C\\
    V_2 & = G_{20}C \oplus G_{11}C \oplus G_{02}C\\
    V_3 & = G_{21}C \oplus G_{12}C\\
    V_4 & = G_{22}C
  \end{aligned}
\end{equation}
  and the corresponding ordered bases.
\end{definition}
Relative to $C = V_0 \oplus \ldots \oplus V_4$
the differential has the form block form
\[
d_C \;=\;
\begin{pmatrix}
  \ast & 0 & 0 & 0 & 0 \\
  \ast & \ast & 0 & 0 & 0 \\
  \ast & \ast & \ast & 0 & 0 \\
  \ast & \ast & \ast & \ast & 0 \\
  \ast & \ast & \ast & \ast & \ast
\end{pmatrix}
\]
based on which we will construct a spectral sequence.
The next lemma gives a first qualitative description of $d_C$,
later we need more details.
\begin{lemma}[$d_C$ is a first order matrix differential operator]
  \label{lemma:fomdo}
  There is a $72 \times 72$ lower block triangular matrix $B$
  and three strictly lower block triangular matrices $A_1,A_2,A_3$,
  all with constant entries,
  such that
  $ d_C = A_1 D_1
          +
          A_2 D_2
          +
          A_3 D_3
          +
          B$.
 \end{lemma}
\begin{proof}
  The diagonal blocks of $A_1,A_2,A_3$ are zero
  since the derivatives $D_1,D_2,D_3$ appear in the Maurer-Cartan element
  with $s_2s_3$, $s_3s_1$, $s_1s_2$.
\qed\end{proof}
\begin{lemma}[Blocks of $d_C$] \label{lemma:blocksdc}
  We denote by
  \begin{align*}
      \ddc{k,p_1p_2p_3 \to q_1q_2q_3}
      \;&\in\;
      \Hom_{\R}(G_{p_1p_2p_3}C^k,G_{q_1q_2q_3}C^{k+1})\\
      \ddc{k,p_2p_3 \to q_2q_3}
      \;&\in\;
      \Hom_{\R}(G_{p_2p_3}C^k,G_{q_2q_3}C^{k+1})\\
      \ddc{k,p \to q}
      \;&\in\;
      \Hom_{\R}(V_p^k,V_q^{k+1})
    \end{align*}
  the blocks of $d_C$ relative to the different gradings of $C$.
  If the index $k$ is omitted, then a direct sum over $k$ is understood.
  We have:
  \begin{itemize}
    \item 
      If $p_i > q_i$ for at least one $i$,
      then the block $\ddc{p_1p_2p_3\to q_1q_2q_3}$ is zero.
    \item Each diagonal block $\ddc{p_1p_2p_3\to p_1p_2p_3}$ 
      is $C^{\infty}(\spatialm,\R)$-linear
      (entries in $C^{\infty}(\spatialm,\R)$).
  \end{itemize}
  Analogous statements hold for
      $\ddc{p_2p_3\to q_2q_3}$
      and $\ddc{p\to q}$. 
\end{lemma}
\begin{proof}
  The linearity over $C^{\infty}(\spatialm,R)$,
  is because $D_1,D_2,D_3$ appear only on the subdiagonals, see \lemmaref{fomdo}.
\qed\end{proof}
\subsection{Spectral sequence: the 0th page}\label{sec:0thpage}
Recall Assumption \ref{assp:mc}.
The 0th-page is
\[
\xymatrix{%
  V_0 \ar@(dl,dr)[]_{D_{00}} &
  V_1 \ar@(dl,dr)[]_{D_{11}} &
  V_2 \ar@(dl,dr)[]_{D_{22}} &
  V_3 \ar@(dl,dr)[]_{D_{33}} &
  V_4 \ar@(dl,dr)[]_{D_{44}}
}
\]
with differentials $D_{pp} = \ddc{p \to p}$ that are linear over
$C^{\infty}(\spatialm,\R)$.
The differentials are also diagonal relative to the decompositions
\eqref{eq:vspaces}, therefore:
\begin{equation}\label{eq:hd01234}
\begin{aligned}
  H(D_{00}) & = H(\ddc{00 \to 00})\\
  H(D_{11}) & = H(\ddc{10 \to 10}) \oplus H(\ddc{01 \to 01})\\
  H(D_{22}) & = H(\ddc{20 \to 20}) \oplus H(\ddc{11 \to 11}) \oplus H(\ddc{02 \to 02})\\
  H(D_{33}) & = H(\ddc{21 \to 21}) \oplus H(\ddc{12 \to 12})\\
  H(D_{44}) & = H(\ddc{22 \to 22})
\end{aligned}
\end{equation}
For example, $D_{11}$ is diagonal because the off-diagonal
$\ddc{10 \to 01}$ and $\ddc{01 \to 10}$
vanish by \lemmaref{blocksdc}.
We now give these differentials in matrix form,
all calculated using the computer.
\begin{itemize}
  \item Consider the block
  $\ddc{00\to 00}$.
  Note that
  \[
    \rank_{C^{\infty}(\spatialm,\R)}
  V_0^k = 5,4,1,0
\]
with ordered bases
  $b_1^{0,000}$, $b_2^{0,000}$, $b_3^{0,000}$, $b_4^{0,000}$,
  $b_5^{0,000}$ respectively
  $b_1^{1,000}$, $b_2^{1,000}$, $b_3^{1,000}$, 
  $b_1^{1,200}$ respectively $b_1^{2,000}$ relative to which
  \[
      \ddc{0,00\to 00}
      =
      \begingroup
      \setlength\arraycolsep{4pt}
      \begin{pmatrix}
        0 & 0 & 0 & 0 & g_1^0\\
        0 & 0 & 0 & 0 & g_2^0\\
        0 & 0 & 0 & 0 & g_3^0\\
          - 2 g_1^0 g_1^1
        & 0
        & 0
        & 0
        & -  g_1^1
      \end{pmatrix}
      \endgroup
  \]
  and
  \[
    \ddc{1,00\to 00}
    \;=\;
    \begin{pmatrix}
      -\tfrac{1}{3}(g_2^0+g_3^0)
      &
      -\tfrac{1}{3}(g_3^0+g_1^0)
      &
      -\tfrac{1}{3}(g_1^0+g_2^0)
      &
      0
    \end{pmatrix}
  \]
  and $\ddc{2,00\to 00} = 0$.
  By construction, $\ddc{00\to 00}$ must itself be a differential.
  To check this explicitly, use the algebraic constraint
  $g_2^0g_3^0 + g_3^0g_1^0 + g_1^0g_2^0 = 0$, see {\citetwo}.
\item  Consider
  $\ddc{10 \to 10}$
  and
  $\ddc{01 \to 01}$.
  Then
  $ \rank_{C^\infty} G_{10}C^k =
  \rank_{C^\infty} G_{01}C^k = 2,8,1,0$ and
  \[
             \ddc{0,10\to 10}
      =
      \begin{pmatrix}
          -1 & 0 \\
          0 & 0 \\
          0 & 0 \\
          0 & 0 \\
          g_1^0 + g_2^0 + g_3^0 & 0 \\
          0 & g_2^0 - g_1^0 \\
          -g_2^0 & 0 \\
          g_1^0 & 0
      \end{pmatrix} 
    \qquad
      \ddc{0,01\to 01}
      =
      \begin{pmatrix}
          -1 & 0 \\
          0 & 0 \\
          0 & 0 \\
          0 & 0 \\
          g_1^0 + g_2^0 + g_3^0 & 0 \\
          0 & g_1^0 - g_3^0 \\
          -g_1^0 & 0 \\
          g_3^0 & 0
      \end{pmatrix}
  \]
  and
  \begin{align*}
        \ddc{1,10 \to 10}
    & =
    \begin{pmatrix}
      0 & 0 & 0 & 0 & -g_1^0-g_2^0 & 0 & g_3^0-g_2^0 & g_1^0-g_3^0
    \end{pmatrix}\\
    \ddc{1,01 \to 01}
    & =
    \begin{pmatrix}
      0 & 0 & 0 & 0 & -g_1^0-g_3^0 & 0 & g_2^0-g_1^0 & g_3^0-g_2^0
    \end{pmatrix}
  \end{align*}
  and $\ddc{2,10\to 10} = 0$ and $\ddc{2,01\to 01} = 0$.
\item
  The block $\ddc{20 \to 20}$ is zero because $\rank_{C^\infty}
  G_{20} C^k = 0,1,0,0$.
  So is $\ddc{02 \to 02}$.
\item The block
  $\ddc{11 \to 11}$ is given by
  \[
      \ddc{0,11\to 11}
      \;=\;
      \begin{pmatrix}
          -1 & 0 \\
          0 & 0 \\
          0 & 0 \\
          0 & 0 \\
          g_1^0 + g_2^0 + g_3^0 & 0 \\
          0 & g_3^0 - g_2^0 \\
          -g_3^0 & 0 \\
          g_2^0 & 0\\
           g_1^1 & 0
      \end{pmatrix}
  \]
  and
    \begin{multline*}
    \ddc{1,11 \to 11}
    \;=\;\\
    \left(
    \begin{array}{c c c c c c c c c}
      0 & 0 & 0 & 0 & -g_2^0-g_3^0 & 0 & g_1^0-g_3^0 & g_2^0-g_1^0 & 0\\
      2  g_1^1 & 0 & 0 & 0 & 0 & 0 & 0 & 0 & 2\\
      0 & 2 g_1^1 & 0 & 0 & 0 & 0 & 0 & 0 & 0 \\
      0 & 0 & 2 g_1^1 & 0 & 0 & 0 & 0 & 0 & 0 \\
      0 & 0 & 0 & 2  g_1^1 & 0 & 0 & 0 & 0 & 0 \\
      0 & 0 & 0 & 0 & - g_1^1 & 0 & 0 & 0 & g_1^0 + g_2^0 + g_3^0 \\
      -2 g_1^0g_1^1 & 0 & 0 & 0 & - g_1^1 & 0 & - g_1^1 &  g_1^1 & -g_1^0
    \end{array}
    \right)
  \end{multline*}
  and $\ddc{2,11\to 11} = 0$.
\item The blocks
  $\ddc{21 \to 21}$ and
  $\ddc{12 \to 12}$ are given by
  \[
      \ddc{1,21\to 21}
      \;=\;
      \begin{pmatrix}
          2 \\
          0  \\
          0  \\
          0  \\
          g_1^0 + g_2^0 + g_3^0 \\
          -g_2^0 
      \end{pmatrix}
      \qquad
      \ddc{1,12\to 12}
      \;=\;
      \begin{pmatrix}
          2 \\
          0  \\
          0  \\
          0  \\
          g_1^0 + g_2^0 + g_3^0 \\
          -g_3^0 
      \end{pmatrix}
    \]
    whereas
    $\ddc{0,21 \to 21} = \ddc{2,21\to 21} = 0$
    and
    $\ddc{0,12 \to 12} = \ddc{2,12\to 12} = 0$.
\item The block $\ddc{22 \to 22}$ is given by
  \[
      \ddc{2,22\to 22}
      \;=\;
      \begin{pmatrix}
          6 \\
          0 \\
          0 \\
          0 \\
          -2(g_1^0 + g_2^0 + g_3^0)
      \end{pmatrix}
    \]
    whereas $\ddc{0,22 \to 22} = \ddc{1,22\to 22} = 0$.
\end{itemize}
\newcommand{\sx}[1]{$\mathrlap{{}_{(#1)}}$}
\begin{lemma}[Homology of the 0th page] \label{lemma:hom0page}
  The $H(\ddc{p_1p_2\to p_1p_2})$
  are free $C^{\infty}(\spatialm,\R)$-modules with ranks:
  \begin{center}
  \begin{tabular}{c | c c c c}
    & $k=0$ & $k=1$ & $k=2$ & $k=3$\\
    \hline
    $H^k(\ddc{00 \to 00})$ & 3 \sx{5} & 1 \sx{4} & \sx{1} & \\
   \hline
   $H^k(\ddc{10 \to 10})$, $H^k(\ddc{01 \to 01})$ & \sx{2} & 5 \sx{8} & \sx{1} & \\
   \hline
   $H^k(\ddc{20 \to 20})$, $H^k(\ddc{02 \to 02})$ & & 1 \sx{1} & & \\
   $H^k(\ddc{11 \to 11})$ & \sx{2} & \sx{9} & \sx{7} & \\
   \hline
   $H^k(\ddc{21 \to 21})$, $H^k(\ddc{12 \to 12})$ & & \sx{1} & 5 \sx{6} & \\
   \hline
   $H^k(\ddc{22 \to 22})$ & & & \sx{1} & 4 \sx{5}
  \end{tabular}
\end{center}
  Accordingly, the $H(D_{pp})$
  are free $C^{\infty}(\spatialm,\R)$-modules with ranks:
  \begin{center}
  \begin{tabular}{c | c c c c}
    & $k=0$ & $k=1$ & $k=2$ & $k=3$\\
    \hline
    $H^k(D_{00})$ & 3 \sx{5} & 1 \sx{4} & \sx{1} & \\
    $H^k(D_{11})$ & \sx{4} & 10 \sx{16} & \sx{2} & \\
    $H^k(D_{22})$ & \sx{2} & 2 \sx{11} & \sx{7}  & \\
    $H^k(D_{33})$ &  & \sx{2} & 10 \sx{12} & \\
    $H^k(D_{44})$ &  &  & \sx{1} & 4 \sx{5}
  \end{tabular}
\end{center}
The small numbers in brackets are the ranks before taking homology.
\end{lemma}
\begin{proof}
  Use Assumption \ref{assp:mc}.
  The claim follows from the $\ddc{p_1p_2 \to p_1p_2}$.
  For freeness it suffices to exhibit bases,
  which are in \lemmaref{hbas}.
\qed\end{proof}
\newcommand{\sppx}[1]{\rule{#1mm}{0pt}}
\begin{lemma}[Ordered bases of representatives]\label{lemma:hbas}
  As always, Assumption \ref{assp:mc} is in force.
  The row vectors in this lemma are relative to the ordered bases
  for the $G_{p_2p_3}C^k$ in \lemmaref{Cspace} and \defref{Cgradings}.
  In each case, an ordered basis for the module to the left of the colon
  is given by
  the vectors to the right of the colon:
\begin{itemize}
  \item $H^0(\ddc{00 \to 00})$:
    $(0,1,0,0,0)$, $(0,0,1,0,0)$, $(0,0,0,1,0)$.
 \item $H^1(\ddc{00 \to 00})$:
    $(g_2^0-g_3^0, g_3^0-g_1^0, g_1^0-g_2^0, 0)$.
  \item
    $H^1(\ddc{10 \to 10})$:
    $(1,0,0,0,0,0,0,0)$,
$(0,1,0,0,0,0,0,0)$,
$(0,0,1,0,0,0,0,0)$,\\
\sppx{20}$(0,0,0,1,0,0,0,0)$,
$(0,0,0,0,0,0,g_1^0-g_3^0,g_2^0-g_3^0)$.
  \item
    $H^1(\ddc{01 \to 01})$:
$(1,0,0,0,0,0,0,0)$,
$(0,1,0,0,0,0,0,0)$,
$(0,0,1,0,0,0,0,0)$,\\
\sppx{20}$(0,0,0,1,0,0,0,0)$,
$(0,0,0,0,0,0,g_2^0-g_3^0,g_2^0-g_1^0)$.
\item 
    $H^1(\ddc{20 \to 20})$
    and
    $H^1(\ddc{02 \to 02})$: $(1)$.
  \item 
    $H^2(\ddc{21 \to 21})$
    and
    $H^2(\ddc{12 \to 12})$:
$(0,1,0,0,0,0)$,
$(0,0,1,0,0,0)$,
$(0,0,0,1,0,0)$,\\
\sppx{45}$(0,0,0,0,1,-1)$,
$(0,0,0,0,0,1)$.
\item
    $H^3(\ddc{22 \to 22})$:
$(0,1,0,0,0)$,
$(0,0,1,0,0)$,
$(0,0,0,1,0)$,
$(0,0,0,0,-\tfrac{2}{3})$.
\end{itemize}
\end{lemma}
\begin{proof}
  By inspection of the matrices for the $\ddc{p_2p_3 \to p_2p_3}$.
\qed\end{proof}

\begin{lemma}[Choice of complements of the kernel] \label{lemma:cbas}
  In $G_{p_2p_3}C^k$ a complement of $\ker \ddc{p_2p_3 \to p_2p_3}$ 
  is spanned by all elements of the form:
  \begin{align*}
    \textup{$G_{00}C^0$:} & \qquad (\ast,0,0,0,\ast)\\
    \textup{$G_{00}C^1$:} & \qquad (1,1,1,0)\\
    \textup{$G_{10}C^1$:} & \qquad (0,0,0,0,\ast,0,0,0)\\
    \textup{$G_{01}C^1$:} & \qquad (0,0,0,0,\ast,0,0,0)\\
    \textup{$G_{11}C^1$:} & \qquad (\ast,\ast,\ast,\ast,\ast,0,\ast,\ast,0)
 \end{align*}
  For all other $G_{p_2p_3}C^k$ the kernel is zero or all,
  so there is a unique complement.
\end{lemma}
\begin{proof}
  By inspection of the matrices for the $\ddc{p_2p_3 \to p_2p_3}$.
\qed\end{proof}
\begin{lemma}[Contractions] \label{lemma:hipdr}
There are unique linear maps
\[
\xymatrix{
  G_{p_2p_3} C \ar@/^/[d]^{p_{p_2p_3}}
               \ar@(ul,ur)[]^{h_{p_2p_3}} \\
  H(\ddc{p_2p_3 \to p_2p_3}) \ar@/^/[u]^{i_{p_2p_3}}
}
\rule{20mm}{0pt}
\xymatrix{
  V_p \ar@/^/[d]^{p_p}
               \ar@(ul,ur)[]^{h_p} \\
  H(\ddc{p \to p}) \ar@/^/[u]^{i_p}
}
\]
such that
$i_{p_2p_3}$ associates to each element of the homology the unique
    representative in the span of the basis in \lemmaref{hbas};
$p_{p_2p_3}$ associates to each element of $\ker d_C$ 
   the corresponding element in the homology,
   and $\ker p_{p_2p_3}$ contains the complements in \lemmaref{cbas};
and $h_{p_2p_3}$ associates to each element of $\image d_C$
   the unique $d_C$-preimage in the complement in \lemmaref{cbas},
   and $\ker h_{p_2p_3}$ contains the elements
   in \lemmaref{hbas} and the complements in \lemmaref{cbas}.
Then $d i = p d = hi = ph = h^2 = 0$
and $pi = \mathbbm{1}$ and $ip = \mathbbm{1} - hd - dh$
with obvious abbreviations.
Analogous statements for
$i_p$, $p_p$, $h_p$. They are block-diagonal
with entries $i_{p_2p_3}$, $p_{p_2p_3}$, $h_{p_2p_3}$
with $p_2+p_3 = p$.
We denote by
$i_{p_2p_3}^k$,
$p_{p_2p_3}^k$,
$i_p^k$,
$p_p^k$
the corresponding maps at homological degree $k$.
\end{lemma}
\begin{proof}
  This is a contraction in standard form.
\qed\end{proof}
\subsection{Spectral sequence: the 1st page} \label{sec:1stpage}
By \secref{0thpage},
and of course with Assumption \ref{assp:mc} always in force,
the 1st page is the disconnected complex,
abbreviating $H^{0,1}(-) = H^0(-) \oplus H^1(-)$,
\[
\xymatrix{%
  H^{0,1} (D_{00})
    \ar[r]_{\Donezero}
    & H^1(D_{11})
    & H^1(D_{22}) \ar[r]_{\Dthreetwo}
    & H^2(D_{33}) \ar[r]_{\Dfourthree}
    & H^3(D_{44})
}
\]
where $\Donezero$
annihilates $H^1(D_{00})$ and we denote by
$\Donezero' : H^0(D_{00}) \to H^1(D_{11})$ the other, nontrivial part.
In block notation using the decompositions \eqref{eq:hd01234} we have
\begin{align*}
  \Donezero'
  &\;=\;
  p_1^1 \ddc{0,0\to 1} i_0^0
  \;=\;
  \begin{pmatrix}
    p_{10}^1 \ddc{0,00 \to 10} i_{00}^0 \\
    p_{01}^1 \ddc{0,00 \to 01} i_{00}^0
  \end{pmatrix}\\
  \Dthreetwo
  &\;=\;
  p_3^2 \ddc{1,2\to 3} i_2^1
  \;=\;
  \begin{pmatrix}
    p_{21}^2 \ddc{1,20 \to 21} i_{20}^1 & 0 & 0 \\
    0 & 0 & p_{12}^2 \ddc{1,02 \to 12} i_{02}^1
  \end{pmatrix}\\
  \Dfourthree
  &\;=\;
  p_4^3 \ddc{2,3\to 4} i_3^2
  \;=\;
  \begin{pmatrix}
    p_{22}^3 \ddc{2,21 \to 22} i_{21}^2 & p_{22}^3 \ddc{2,12 \to 22} i_{12}^2
  \end{pmatrix} 
\end{align*}
The leftmost and rightmost zeros in $\Dthreetwo$
are due to $\ddc{20\to 12} = \ddc{02 \to 21} = 0$,
the zeros in the middle are due to $i_{11}=0$,
equivalently
$H(\ddc{11 \to 11}) = 0$.

The spaces are free $C^{\infty}(\spatialm,\R)$-modules with bases chosen,
and $\Donezero'$, $\Dthreetwo$, $\Dfourthree$ are matrix differential operators.
By direct calculation,
\[
  \Donezero'\;=\;
  \begin{pmatrix}
 0 & 0 & 0 \\
 D_3 & 2 g_1^1 & 0 \\
 -2 g_2^1 & D_3 & 0 \\
 0 & 0 & D_3 \\
 0 & 0 & 0 \\
 0 & 0 & 0 \\
 D_2 & 0 & -2g_1^1\\
 0 & D_2 & 0 \\
 2g_3^1 & 0 & D_2 \\
 0 & 0 & 0
 \end{pmatrix}
 \rule{30pt}{0pt}
  \Dthreetwo
  \;=\;
  \begin{pmatrix}
    0 & 0 \\
    2 & 0 \\
    0 & 0 \\
    0 & 0 \\
    -D_2 & 0 \\
    0 & 0 \\
    0 & 0 \\
    0 & 2 \\
    0 & 0\\
    0 & - D_3
  \end{pmatrix}
\]
and
\[
  \Dfourthree\;=\;
  \begin{pmatrix}
    D_2 & 0 & -2g_1^1 & 0 & 0 & D_3 & 2g_1^1 & 0 & 0 & 0 \\
    0 & D_2 & 0 & 0 & 2 & -2g_2^1 & D_3 & 0 & 0 & 0 \\
    2g_3^1 & 0 & D_2 & 0 & 0 & 0 & 0 & D_3 & 0 & 2 \\
    0 & 0 & 0 & D_2 & 0 & 0 & 0 & 0 & D_3 & 0
  \end{pmatrix}
\]
Necessarily $\Dfourthree\Dthreetwo=0$, which can also be checked directly.

\begin{lemma} \label{lemma:1stpagehomology} \rule{0pt}{0pt}
  \begin{itemize}
    \item $\ker \Donezero' \simeq \R^3$ as vector spaces.
    \item $\ker \Dthreetwo = 0$.
    \item $\coker \Dfourthree = 0$.
  \end{itemize}
\end{lemma}
\begin{proof}
  By inspection of their matrices.
  For $\Donezero'$
   note that the $D_i$ are the left-invariant
    vector fields on an open subset of a Lie group; 
every left-invariant vector field commutes with
every right-invariant vector field;
and the space of right-invariant vector fields is $\simeq \R^3$.
  The injectivity of $\Dthreetwo$ is clear.
  The map $\Dfourthree$ is surjective
  because already its leftmost $4 \times 4$
  block is surjective
  by \eqref{eq:d2cond}
  and the global existence theorem for linear ordinary differential
  equations.
\qed\end{proof}
\subsection{Spectral sequence: the 2nd page} \label{sec:2ndpage}
By \secref{1stpage},
and of course with Assumption \ref{assp:mc} always in force,
the 2nd page is the following complex,
with $\R^3$ in homological degree zero:
\[
\xymatrix{%
  \R^3 \oplus H^1(D_{00})
  & \frac{H^1(D_{11})}{\image \Donezero}
    \ar@/_4mm/[rr]_{\Dthreeone}
  & \jj{0}
  & \frac{\ker \Dfourthree}{\image \Dthreetwo}
  & \jj{0}
}
\]
Our main goal is to show that $\Dthreeone$ is surjective.
By the way the spectral sequence is constructed, $\Dthreeone$ is induced by,
with $h_2$ as in \lemmaref{hipdr}:
\[
 \ddc{1 \to 3}
  - \ddc{2\to 3} h_2 \ddc{1\to 2}
  \;:\; V_1 \to V_3
\]
The last map induces a map
$H(D_{11}) \to H(D_{33})$,
because  $H(D_{22})=0$
as witnessed by $h_2$ via $h_2 D_{22} + D_{22} h_2 = \mathbbm{1}$.
More specifically, it induces a map
$\delta : H^1(D_{11}) \to H^2(D_{33})$
with $\image \delta \subset \ker \Dfourthree$.
We emphasize $\delta$
because it is a map between free $C^{\infty}(\spatialm,\R)$-modules
for which bases have been chosen,
and it can be written as a matrix differential operator,
of size $10 \times 10$.
Clearly
\begin{equation}\label{eq:d13eq}
  \coker \Dthreeone = 0
  \qquad
  \Longleftrightarrow
  \qquad
  \image \delta
  \;+\;
  \image \Dthreetwo \;=\; \ker \Dfourthree
\end{equation}
In block notation using the decompositions \eqref{eq:hd01234} we have
\begin{multline*}
  \delta
  \;=\;
  \begin{pmatrix} p_{21}^2 & 0 \\ 0 & p_{12}^2 \end{pmatrix}
  \begin{pmatrix} \ddc{1,10\to 21} & \ddc{1,01\to 21}\\
                  \ddc{1,10\to 12} & \ddc{1,01\to 12}
                \end{pmatrix}
  \begin{pmatrix} i_{10}^1 & 0 \\ 0 & i_{01}^1 \end{pmatrix}
  \\-
  \begin{pmatrix} p_{21}^2 & 0 \\ 0 & p_{12}^2 \end{pmatrix}
  \begin{pmatrix} \ddc{1,11\to 21} \\ \ddc{1,11\to 12} \end{pmatrix}
  h_{11}^2
  \begin{pmatrix}
                  \ddc{1,10\to 11} & \ddc{1,01\to 11}
                \end{pmatrix}
  \begin{pmatrix} i_{10}^1 & 0 \\ 0 & i_{01}^1 \end{pmatrix}
\end{multline*}
where $h_{11}^2$ is that part of $h_{11}$ in
\lemmaref{hipdr} that maps from
homological degree $k=2$ back to $k=1$, a $9 \times 7$ matrix with constant entries.
Now:
\begin{itemize}
  \item Let $\slashedDfourthree$ be the $1 \times 5$ matrix operator obtained
from $\Dfourthree$ by deleting rows $1$, $2$, $3$ and columns $2$, $3$, $5$, $8$, $10$.
Explicitly
$\slashedDfourthree
  =
  \begin{pmatrix} 0 & D_2 & 0 & 0 & D_3 \end{pmatrix}$.
\item 
Let $\slashed{\delta}$ be the $5 \times 10$ matrix obtained
from $\delta$ by deleting rows $2$, $3$, $5$, $8$, $10$.
\end{itemize}
By construction
$\image \slashed{\delta} \subset \ker \slashedDfourthree$.
We claim that
\begin{equation} \label{eq:d13eq2}
  \coker \Dthreeone = 0
  \qquad
  \Longleftrightarrow
  \qquad
  \image \slashed{\delta}
  \;=\; \ker \slashedDfourthree
\end{equation}
This follows from \eqref{eq:d13eq}
by deriving a sequence of equivalent conditions, as follows.
Column $5$ of $\Dfourthree$
has a single nonzero entry, in row $2$, so we get an equivalent condition if
we delete column $5$ and row $2$.
Similarly column $10$ and row $3$.
Column $3$ now only contains a single nonzero non-deleted entry, in row $1$,
so we can also delete
column $3$ and row $1$.
Accordingly we must delete rows $3$, $5$, $10$ in $\delta$ and $\Dthreetwo$.
Given the structure of $\Dthreetwo$, we
can now delete columns $2$ and $8$ in $\Dfourthree$,
accordingly rows $2$ and $8$ in $\delta$ and $\Dthreetwo$. 
No nonzero non-deleted entry is left in $\Dthreetwo$ and the claim follows.

Explicitly,
\begin{multline*}
  \slashed{\delta}\;=\;\\
  \left(
  \begingroup
  \setlength\arraycolsep{4pt}
  \begin{array}{c c c c c}
    g_1^0 + g_2^0 + 2g_3^0 & -2D_1 - \tfrac{D_2D_3}{2g_1^1}
    & -D_2 & D_3 & g_1^0-g_2^0 \\
    \tfrac{g_1^0+g_2^0+g_3^0}{4g_1^1}
    (D_3D_2 + 2g_1^1 D_1) & 0 & 0 & 0 & 0 \\
    0 & \tfrac{D_2D_2}{2g_1^1} & 0 & -2D_2 & 0 \\
    \tfrac{(g_1^0)^2 + g_1^0g_3^0 + (g_3^0)^2}{g_1^1(g_2^0-g_3^0)} D_2 
     & 0 & \tfrac{D_2D_2}{2g_1^1} + 2g_3^1 & 0 & 0 \\
    -\tfrac{g_1^0+g_2^0+g_3^0}{4g_1^1}
    (D_2D_2 + 4g_1^1 g_3^1) & 0 & 0 & 0 & 0 \\
  \end{array}\endgroup\right. \displaybreak[0]\\
  \left.
  \begingroup
  \setlength\arraycolsep{4pt}
  \begin{array}{c c c c c}
    0 & \tfrac{D_3D_3}{2g_1^1} & 2D_3 & 0 & 0 \\
    -\tfrac{g_1^0+g_2^0+g_3^0}{4g_1^1}
    (D_3D_3 + 4g_1^1 g_2^1) & 0 & 0 & 0 & 0 \\
    -(g_1^0 + 2g_2^0 + g_3^0) & D_1 - \tfrac{D_2D_3}{2g_1^1}
    & -D_2 & D_3 & g_3^0-g_1^0 \\
    -\tfrac{2(g_1^0)^2 + (g_3^0)^2}{2g_1^1(g_2^0-g_3^0)} D_3
    & \frac{g_2^1}{g_1^1} D_2 & D_1 - \frac{D_2D_3}{2g_1^1} & 2g_2^1 & \frac{g_1^0-g_3^0}{2g_1^1}D_3 \\
    \tfrac{g_1^0+g_2^0+g_3^0}{4g_1^1}
    (D_2D_3 - 2g_1^1 D_1) & 0 & 0 & 0 & 0
  \end{array}\endgroup\right)
\end{multline*}

Consider the submatrix of $\slashed{\delta}$ 
containing only columns $4$, $5$, $9$.
Using Assumption \ref{assp:mc},
specifically $g_2^1 \neq 0$ 
and \eqref{eq:d2cond} and $g_1^0 - g_2^0 \neq 0$,
this submatrix' image is already all
vectors of the form $(\ast,0,\ast,\ast,0)$.
Therefore we get a condition equivalent to \eqref{eq:d13eq2} by deleting
columns $1$, $3$, $4$ of $\slashedDfourthree$
and rows $1$, $3$, $4$ of $\slashed{\delta}$.
Only a $2\times 2$
submatrix of the
non-deleted part of $\slashed{\delta}$
remains, all other entries are zero.
We get
\begin{multline} \label{eq:2by2sys}
  \coker \Dthreeone = 0
  \qquad
  \Longleftrightarrow
  \\
  \image
  \begin{pmatrix}
    D_3D_2 + 2g_1^1 D_1 & -D_3D_3 - 4 g_1^1 g_2^1\\
    -D_2D_2 - 4g_1^1 g_3^1 & D_2D_3 - 2g_1^1 D_1
  \end{pmatrix}
  \;=\; \ker \begin{pmatrix} D_2 & D_3 \end{pmatrix}
\end{multline}
Here $\subset$ is by construction,
so $\supset$ is what we show.
We show that the first column of the $2 \times 2$ matrix
suffices,
using coordinates as in \eqref{eq:d2cond}.
Given any $(f_2,f_3) \in \ker (D_2\, D_3)$,
it suffices to find a function $h$ such that
\begin{align*}
  ((D_3D_2+2g_1^1D_1)h)|_\Sigma & = f_2|_\Sigma\\
(-D_2D_2-4g_1^1 g_3^1)h & = f_3
\end{align*}
where $\Sigma = \{x^2 = 0\}$.
The second equation has a unique solution for every choice
of $h|_\Sigma$ and $(\p_2 h)|_\Sigma$ by \eqref{eq:d2cond},
and we fix $(\p_2 h)|_\Sigma = 0$.
The first equation is now equivalent to
$2g_1^1 \p_1 h = f_2 + a_2(f_3 + 4g_1^1 g_3^1 h)$
as an equation on $\Sigma$ for $h|_\Sigma$,
with $a_2$ as in $D_3 = a_1 \p_1 + a_2 \p_2 + a_3 \p_3$.
Since $g_1^1 \neq 0$, a solution $h|_\Sigma$ exists,
again by \eqref{eq:d2cond}.
We have proved:
\begin{lemma}
  $\coker \Dthreeone = 0$.
\end{lemma}
By similar arguments
we have
$\image \Donezero' \subset \ker \slashed{\delta}$, which one can also
check directly
using the explicit formulas given before.
Every $x$ satisfies $\delta x \in \image \Dthreetwo$
if and only if $\slashed{\delta}x = 0$;
here $\Rightarrow$ is trivial,
$\Leftarrow$ uses $\image \delta \subset \ker \Dfourthree$.
Hence (see \theoremref{noobstr}):
\begin{lemma} \label{lemma:kerd13}
  $\ker \Dthreeone = \ker \slashed{\delta} / \image \Donezero'$. 
\end{lemma}

This could be analyzed further.
Informally counting $C^{\infty}(\spatialm,\R)$-ranks,
and knowing $\image \slashed{\delta} \subset \ker \slashedDfourthree$,
we expect that $\ker \slashed{\delta}$
has rank $10 -5 + 1 = 6$,
and $\image \Donezero'$ has rank $3$,
so $\ker \Dthreeone$ should have rank $6-3 = 3$.
But these are not free $C^{\infty}(\spatialm,\R)$-modules,
so the counting does not really make sense,
and in any case is modulo lower dimensional data.
One can write down a parametrization of $\ker \Dthreeone$,
but we will not.

\section{The homology in $\abounce$\\
of some homogeneous Maurer-Cartan elements}\label{sec:homologybounce}
This section is closely analogous to \secref{homology}.
In fact, the differences are so small that we only say what 
has to be modified relative to \secref{homology}.
For a homogeneous $\gamma \in \MC(\abounce)$
we compute the homologies $H^k(d)$ of
\[
  d = [\gamma,-]\;\;\in\;\; \End^1(\abounce)
\]

\begin{assumption}[In force through \secref{homologybounce}] \label{assp:mcbounce}
We assume
\[
  \gamma \in \MC(\abounce) \cap (\theta_0 D_0 + {\abounce^1}_{,G})
\]
is given,
using the over-parametrization in Lemma \depref{npar}, by
\begin{align*}
\mu_1  \;&=\;  -\tfrac{1}{2} \log(2\cosh t) &
\gamma_1^1  \;&=\;  1\\
\mu_2  \;&=\;  -\tfrac{tu}{2} + \tfrac{1}{2} \log(2\cosh t) &
\gamma_2^1  \;&=\;  g_2^1\\
\mu_3  \;&=\;  -\tfrac{t}{2u} + \tfrac{1}{2} \log(2\cosh t) &
\gamma_3^1  \;&=\;  g_3^1\\
\gamma_1^0  \;&=\;  \tfrac{1}{2}-\chi &
(\beta_1,\beta_2,\beta_3) \;&=\;  (D_1,D_2,D_3)\\
\gamma_2^0  \;&=\;  -\tfrac{1}{2} (1+u)+\chi &
\gamma_{1,2,3}^{2,3,4,5,6}  \;&=\; 0 \\
\gamma_3^0  \;&=\;  -\tfrac{1}{2} (1+\tfrac{1}{u})+\chi
\end{align*}
where $\transition = \tfrac{1}{2}(1+\tanh t)$,
a special case of the solution 
in Theorem \depref{bouncesol}, with:
  \begin{itemize}
    \item $u>0$ is constant and $u \neq 1$.
    \item $g_2^1$ and $g_3^1$ are constant and nonzero.
    \item $\spatialm$ is an open subset of a smooth Lie group
      containing the identity element,
      and $D_1$, $D_2$, $D_3$ is a basis of left-invariant vector fields,
      with $[D_2,D_3] = -2D_1$ and $[D_3,D_1] = -2g_2^1 D_2$ and $[D_1,D_2] = -2g_3^1D_3$.
    \item There is an $\eps > 0$ and a diffeomorphism
      $(x^1,x^2,x^3) : \spatialm \to (-\eps,\eps)^3$
      such that $(0,0,0)$ is the identity element and,
      with $\p_1,\p_2,\p_3$ the partial
      derivatives for the coordinate system $x^1,x^2,x^3$,
      we have
      \begin{equation}\label{eq:d2condbounce}
        D_2 = \p_2
        \qquad
        D_1|_{x^2 = 0} = \p_1
      \end{equation}
  \end{itemize}
  \end{assumption}
  \begin{theorem}[No obstructions] \label{theorem:noobstrbounce}
    With Assumption \ref{assp:mcbounce}
    we have:
    \begin{itemize}
      \item $H^0(d) \simeq \R^3$.
      \item $\Gr H^1(d) \simeq H^1(D_{00}) \oplus \ker \Dthreeone$,
        with symbols defined afresh,
        see Lemmas \ref{lemma:hom0pagebounce}, \ref{lemma:kerd13bounce}.
      \item $H^2(d) = 0$.
      \item $H^3(d) = 0$.
      \item $H^4(d) = 0$.
    \end{itemize}
    Here $\simeq$ is an isomorphism as vector spaces,
    and $\Gr$ is the associated graded for the decreasing filtration
    coming from
    the $\Z$-grading of $\abounce$ by $p_2 + p_3$.
  \end{theorem}
  \begin{proof}
    A spectral sequence calculation analogous to that for \theoremref{noobstr}.
  \qed\end{proof}

\subsection{Overview}
The analogy with \secref{homology} is via the module identification
$\abounce \simeq \mathcal{U} \simeq \afree$,
see {\citetwo},
and the fact that $p_1$ played a passive role.
The bracket on $\abounce$ only respects $p_2p_3$.
The general structure of the spectral sequence is the same.
The differential $d$ is new,
and so are the objects derived from it.
Beware that we will use the same symbols as before,
even though they have a different meaning.

\subsection{Contraction and the complex $(C,d_C)$}

\begin{lemma}[The space $C$ and basis elements] \label{lemma:Cspacebounce}
  Define
  \[
    C \;=\; \mathcal{U}_G/t\mathcal{U}_G
    \;\simeq \;{\abounce}_{,G} / t {\abounce}_{,G}
  \]
  A basis is given as in \lemmaref{Cspace}.
  In this section we set $s_1 = 1$ because the bracket on $\abounce$
  does not respect the grading that $s_1$ keeps track of.
\end{lemma}
\begin{lemma}[The contraction and the differential $d_C$] \label{lemma:defretbounce}
  Like \lemmaref{defret}, with
\[
  w\;:\; {\abounce}_{,G} \hookrightarrow \abounce 
  \xrightarrow{d} \abounce \twoheadrightarrow \abounce/{\abounce}_{,G}
\]
\end{lemma}
\begin{example}
  In comparison to \secref{homology}, it is a little harder
  to explicitly write down elements of $\ker w$.
  However this is not required to calculate the differential $d_C$,
  it suffices to make finite order Taylor expansions of elements of $\ker w$ about $t=0$.
  \end{example}
\begin{definition}[Other gradings of $C$] \label{def:Cgradingsbounce}
  Like Definition \ref{def:Cgradings}.
\end{definition}
\begin{lemma}[$d_C$ is a first order matrix differential operator] \label{lemma:fomdobounce}
  Like \lemmaref{fomdo}.
\end{lemma}
\begin{lemma}[Blocks of $d_C$] \label{lemma:blocksdcbounce}
  Define the blocks $\ddc{p_2p_3\to q_2q_3}$
      and $\ddc{p\to q}$ just like in \lemmaref{blocksdc}.
  The same statements hold here.
\end{lemma}
\subsection{Spectral sequence: the 0th page}\label{sec:0thpagebounce}
Analogously to \secref{0thpage}, we have:
\begin{itemize}
  \item The block
  $\ddc{00\to 00}$ is given by
  \[
      \ddc{0,00\to 00}
      =
      \begingroup
      \setlength\arraycolsep{4pt}
      \begin{pmatrix}
        -\tfrac{1}{2} & 0 & 0 & 0 & 0\\
         \tfrac{1}{2} & 0 & 0 & 0 & -\tfrac{u}{2}\\
         \tfrac{1}{2} & 0 & 0 & 0 & -\tfrac{1}{2u}\\
        0 & 0 & 0 & 0 & - \tfrac{1}{2}
      \end{pmatrix}
      \endgroup
  \]
  and
  $\ddc{1,00\to 00}
    =
    \begin{pmatrix}
      \tfrac{1+u^2}{6u}
      &
      \tfrac{1}{6u}
      &
      \tfrac{u}{6}
      &
      -\tfrac{1}{3}
    \end{pmatrix}$
  and $\ddc{2,00\to 00} = 0$.
\item  The blocks
  $\ddc{10 \to 10}$
  and
  $\ddc{01 \to 01}$ are given by
  \[
             \ddc{0,10\to 10}
      =
      \begin{pmatrix}
          -1 & 0 \\
          0 & 0 \\
          0 & 0 \\
          0 & 0 \\
          -\tfrac{1+u^2}{2u} & 0 \\
          -\tfrac{1}{2} & -\tfrac{u}{2} \\
           \tfrac{u}{2} & -1 \\
                      0 & -1
      \end{pmatrix} 
    \qquad
      \ddc{0,01\to 01}
      =
      \begin{pmatrix}
          -1 & 0 \\
          0 & 0 \\
          0 & 0 \\
          0 & 0 \\
          -\tfrac{1+u^2}{2u} & 0 \\
          \tfrac{1}{2} & \tfrac{1}{2u} \\
           0 & 1 \\
           -\tfrac{1}{2u} & 1
      \end{pmatrix} 
  \]
  and
  \begin{align*}
        \ddc{1,10 \to 10}
    & =
    \begin{pmatrix}
      0 & 0 & 0 & 0 & \tfrac{u}{2} & -1 & \tfrac{-1+u^2}{2u} & \tfrac{1}{2u}
    \end{pmatrix}\\
    \ddc{1,01 \to 01}
    & =
    \begin{pmatrix}
      0 & 0 & 0 & 0 & \tfrac{1}{2u} & 1 & -\tfrac{u}{2} & \tfrac{-1+u^2}{2u}
    \end{pmatrix}
  \end{align*}
  and $\ddc{2,10\to 10} = 0$ and $\ddc{2,01\to 01} = 0$.
\item
  The block $\ddc{20 \to 20}$ is zero because $\rank_{C^{\infty}}
  G_{20} C^k = 0,1,0,0$.\\
  The block $\ddc{02 \to 02}$ is zero because $\rank_{C^{\infty}}
  G_{02} C^k = 0,1,0,0$.
\item The block
  $\ddc{11 \to 11}$ is given by
  \[
      \ddc{0,11\to 11}
      =
      \begin{pmatrix}
          -1 & 0 \\
          0 & 0 \\
          0 & 0 \\
          0 & 0 \\
          -\tfrac{1+u^2}{2u} & 0 \\
          0 & \tfrac{-1+u^2}{2u} \\
          \tfrac{1}{2u} & 0 \\
          -\tfrac{u}{2} & 0\\
          \tfrac{1}{2} & 0
      \end{pmatrix}
    \quad
    \ddc{1,11 \to 11}
    =
    \begin{pmatrix}
     -1 & 0 & 0 & 0 & \tfrac{1+u^2}{2u} & 0 & \tfrac{1}{2u} & -\tfrac{u}{2} & -1\\
      1 & 0 & 0 & 0 & 0 & 0 & 0 & 0 & 2\\
      0 & 1 & 0 & 0 & 0 & 0 & 0 & 0 & 0\\
      0 & 0 & 1 & 0 & 0 & 0 & 0 & 0 & 0\\
      0 & 0 & 0 & 1 & 0 & 0 & 0 & 0 & 0\\
      0 & 0 & 0 & 0 & -\tfrac{1}{2} & 0 & 0 & 0 & -\tfrac{1+u^2}{2u}\\
      0 & 0 & 0 & 0 & -\tfrac{1}{2} & 0 & -\tfrac{1}{2} & \tfrac{1}{2} & 0
    \end{pmatrix}
  \]
  and $\ddc{2,11\to 11} = 0$.
\item The blocks
  $\ddc{21 \to 21}$ and
  $\ddc{12 \to 12}$ are given by
  \[
      \ddc{1,21\to 21}
      \;=\;
      \begin{pmatrix}
          2 \\
          0  \\
          0  \\
          0  \\
          -\tfrac{1+u^2}{2u} \\
          \tfrac{1}{2u}
      \end{pmatrix}
      \qquad
      \ddc{1,12\to 12}
      \;=\;
      \begin{pmatrix}
          2 \\
          0  \\
          0  \\
          0  \\
          -\tfrac{1+u^2}{2u} \\
          \tfrac{u}{2}
      \end{pmatrix}
    \]
    whereas
    $\ddc{0,21 \to 21} = \ddc{2,21\to 21} = 0$
    and
    $\ddc{0,12 \to 12} = \ddc{2,12\to 12} = 0$.
\item The block $\ddc{22 \to 22}$ is given by
  \[
      \ddc{2,22\to 22}
      \;=\;
      \begin{pmatrix}
          6 \\
          0 \\
          0 \\
          0 \\
          \tfrac{1}{u} + u
      \end{pmatrix}
    \]
    whereas $\ddc{0,22 \to 22} = \ddc{1,22\to 22} = 0$.
\end{itemize}

\begin{lemma}[Homology of the 0th page] \label{lemma:hom0pagebounce}
  Like \lemmaref{hom0page}.
\end{lemma}
\begin{lemma}[Ordered bases of representatives]\label{lemma:hbasbounce}
  \rule{0pt}{0pt}
\begin{itemize}
  \item $H^0(\ddc{00 \to 00})$:
    $(0,1,0,0,0)$, 
    $(0,0,1,0,0)$,
    $(0,0,0,1,0)$.
 \item $H^1(\ddc{00 \to 00})$:
    $(\tfrac{2u}{1+u^2},0,0,1)$.
  \item 
    $H^1(\ddc{10 \to 10})$:
    $(0,0,0,0,-\tfrac{1}{u^2},0,0,1)$,
$(0,1,0,0,0,0,0,0)$,
$(0,0,1,0,0,0,0,0)$,\\
\sppx{25}$(0,0,0,1,0,0,0,0)$,
$(0,0,0,0,0,0,1,1-u^2)$.
  \item 
    $H^1(\ddc{01 \to 01})$:
$(0,0,0,0,1-u^2,0,0,1)$,
$(0,1,0,0,0,0,0,0)$,
$(0,0,1,0,0,0,0,0)$,\\
\sppx{25}$(0,0,0,1,0,0,0,0)$,
$(0,0,0,0,0,0,1,\tfrac{u^2}{-1+u^2})$.
\item $H^1(\ddc{20 \to 20})$ and
    $H^1(\ddc{02 \to 02})$: $(1)$.
  \item 
    $H^2(\ddc{21 \to 21})$:
$(0,1,0,0,0,0)$,
$(0,0,1,0,0,0)$,
$(0,0,0,1,0,0)$,\\
\sppx{25}$(1+\tfrac{1}{u^2},0,0,0,-\tfrac{1+u^2}{4u^3},0)$,
$(-\tfrac{4}{u},0,0,0,1+\tfrac{1}{u^2},0)$.
  \item 
    $H^2(\ddc{12 \to 12})$:
$(0,1,0,0,0,0)$,
$(0,0,1,0,0,0)$,
$(0,0,0,1,0,0)$,\\
\sppx{25}$(1+u^2,0,0,0,-\tfrac{1}{4} u(1+u^2),0)$,
$(-4u,0,0,0,1+u^2,0)$.
\item
    $H^3(\ddc{22 \to 22})$:
$(0,1,0,0,0)$
$(0,0,1,0,0)$,
$(0,0,0,1,0)$,
$(1,0,0,0,0)$.
\end{itemize}
\end{lemma}

\begin{lemma}[Choice of complements of the kernel] \label{lemma:cbasbounce}
  Like \lemmaref{cbas}.
\end{lemma}
\begin{lemma}[Contractions] \label{lemma:hipdrbounce}
  Like \lemmaref{hipdr}.
\end{lemma}
\subsection{Spectral sequence: the 1st page} \label{sec:1stpagebounce}
Like \secref{1stpage},
but here we work relative to new bases.
Nevertheless, things have been arranged so that:
\begin{lemma}
  The matrices for
  $\Donezero'$, $\Dthreetwo$, $\Dfourthree$ are exactly as in \secref{1stpage},
  with the understanding that $g_1^1$ must be replaced by $1$.
\end{lemma}
\begin{lemma} \label{lemma:1stpagehomologybounce}
  Like \lemmaref{1stpagehomology}.
\end{lemma}
\subsection{Spectral sequence: the 2nd page} \label{sec:2ndpagebounce}
Like \secref{2ndpage}.
We claim that $\image \slashed{\delta} \subset \ker \slashedDfourthree$ and
\eqref{eq:d13eq2} still hold.
The formula for $\slashedDfourthree$ remains the same.
Here however
\begin{multline*}
  \slashed{\delta}\;=\;
  \left(
  \begingroup
  \setlength\arraycolsep{4pt}
  \begin{array}{c c c c c}
    -\tfrac{2(2+u^2)}{u^2} & -4D_1 - D_2D_3
    & -2D_2 & 2D_3 & 2u^2 \\
    -\tfrac{2}{u(1+u^2)}
    (D_3D_2 + 2D_1) & 0 & 0 & 0 & 0 \\
    0 & D_2D_2 & 0 & -4D_2 & 0 \\
    \tfrac{2}{u^2(1-u^2)} D_2 
     & 0 & D_2D_2 + 4g_3^1 & 0 & 0 \\
     \tfrac{2}{u(1+u^2)} 
    (D_2D_2 + 4g_3^1) & 0 & 0 & 0 & 0 \\
  \end{array}\endgroup\right. \displaybreak[0]\\
  \left.
  \begingroup
  \setlength\arraycolsep{4pt}
  \begin{array}{c c c c c}
    0 & D_3D_3 & 4D_3 & 0 & 0 \\
    -\tfrac{2u(1-u^2)}{1+u^2} 
    (D_3D_3 + 4 g_2^1) & 0 & 0 & 0 & 0 \\
    2(2u^2-1) & 2D_1 - D_2D_3
    & -2D_2 & 2D_3 & -\tfrac{2}{1-u^2} \\
    D_3
    & 2 g_2^1 D_2 & 2D_1 - D_2D_3 & 4g_2^1 & \frac{1}{1-u^2} D_3 \\
    \tfrac{2u(1-u^2)}{1+u^2}
    (D_2D_3 - 2 D_1) & 0 & 0 & 0 & 0
  \end{array}\endgroup\right)
\end{multline*}
By the same arguments, we again get
\eqref{eq:2by2sys} and then:
\begin{lemma}
  $\coker \Dthreeone = 0$.
\end{lemma}
\begin{lemma} \label{lemma:kerd13bounce}
  $\ker \Dthreeone = \ker \slashed{\delta} / \image \Donezero'$.
\end{lemma}

\section{Acknowledgements}
This work was done while M.R.~was at ETH Zurich, Switzerland.

\newcommand{\href}[2]{#2}


\begin{thebibliography}{}
\bibitem{rt1}
Reiterer M.~and Trubowitz E.,
\href{https://arxiv.org/abs/1812.11487}{arxiv.org/abs/1812.11487 (2018)}\\
\emph{The graded Lie algebra of general relativity}
\bibitem{rt2}
Reiterer M.~and Trubowitz E.,
\href{https://arxiv.org/abs/1905.09026}{arxiv.org/abs/1905.09026 (2019)}\\
\emph{Filtered expansions in general relativity I}
    \bibitem{fil}
Reiterer M.~and Trubowitz E., arxiv.org/abs/1505.06662 (2015)\\
\emph{Filtered expansions in general relativity and one BKL-bounce}
\bibitem{Gerstenhaber}
Gerstenhaber M., Ann.~of Math.~79 (1964), 59-103\\
\emph{On the deformation of rings and algebras}
  \bibitem{bkl}
    Lifshitz E.M.~and Khalatnikov I.M., Adv.~Phys.~12 (1963), 185-249\\
    \emph{Investigations in relativistic cosmology}\\
    Belinskii V.A., Khalatnikov I.M.~and Lifshitz E.M., Adv.~Phys.~19 (1970), 525-573\\
    \emph{Oscillatory Approach to a Singular Point in the Relativistic Cosmology}\\
Belinskii V.A., Khalatnikov I.M.~and Lifshitz E.M., Adv.~Phys.~31 (1982), 639-667\\
\emph{A general solution of the Einstein equations with a time singularity}
\end{thebibliography}
\end{document}